\documentclass[11pt]{article}
\usepackage{amsmath,amsthm}
\usepackage[hidelinks]{hyperref}
\usepackage{enumitem}

\def\reusesizes{
\usepackage[a4paper]{geometry}
\setlength{\textheight}{19.8cm}
\setlength{\textwidth}{13.5cm}
\setlength{\oddsidemargin}{1.2cm}
\setlength{\evensidemargin}{1.2cm}
\setlength{\topmargin}{1.65cm} 
}
\reusesizes

\usepackage[small]{caption}
\usepackage[subrefformat=parens]{subcaption}
\usepackage{tikz}
\usetikzlibrary{calc}
\usetikzlibrary{shapes.multipart,positioning}

\newtheorem{theorem}{Theorem}
\newtheorem{lemma}[theorem]{Lemma}
\newtheorem{corollary}[theorem]{Corollary}
\newtheorem{proposition}[theorem]{Proposition}
\theoremstyle{definition}
\newtheorem{example}[theorem]{Example}

\newtheorem{definition}[theorem]{Definition}
\newtheorem{problem}{Problem}
\newtheorem{claim}[theorem]{Claim}

\usepackage[lined,ruled]{algorithm2e}

\newcommand{\red}{\textup{red}}

\title{Computing shortest 12-representants of labeled graphs}
\author{Asahi~Takaoka\thanks{
  College of Information and Systems, 
  Muroran Institute of Technology, 
  Muroran-shi, Hokkaido, 050--8585 Japan. 
  \textit{E-mail}: 
  \texttt{\href{mailto:takaoka@mmm.muroran-it.ac.jp}{takaoka@mmm.muroran-it.ac.jp}}
}}
\date{\today}

\begin{document}
\maketitle
\begin{abstract}
The notion of $12$-representable graphs was introduced as a variant of 
a well-known class of word-representable graphs. 
Recently, these graphs were shown to be equivalent to 
the complements of simple-triangle graphs.
This indicates that a $12$-representant of a graph 
(i.e., a word representing the graph) can be obtained 
in polynomial time if it exists. 
However, the $12$-representant is not necessarily optimal (i.e., shortest possible). 
This paper proposes an $O(n^2)$-time algorithm to generate 
a shortest $12$-representant of a labeled graph, 
where $n$ is the number of vertices of the graph. 

\end{abstract}

\section{Introduction}\label{sec:introduction}
A graph $G$ is \emph{word-representable} 
if there is a word $w$ over the alphabet $V(G)$ 
such that two letters $x$ and $y$ are adjacent in $G$ 
if and only if a word $xyxy \cdots$ or a word $yxyx \cdots$ remains 
after removing all other letters from $w$. 
Such a word $w$ is called a \emph{word-representant} of $G$. 
Recently, word-representable graphs have been investigated intensively~\cite{Kitaev17-LNCS,KL15-book}. 
One of the results relevant to this paper is 
the NP-hardness of the recognition, 
see~\cite[Theorem 39]{Kitaev17-LNCS} or~\cite[Theorem 4.2.15]{KL15-book}. 

Jones et al.~\cite{JKPR15-EJC} introduced 
the notion of $u$-representable graphs 
as a generalization of word-representable graphs. 
In this context, word-representable graphs are
called $11$-representable graphs. 
Kitaev~\cite{Kitaev17-JGT} showed that 
only two graph classes are nontrivial
in the theory of $u$-representable graphs: 
$11$-representable graphs and $12$-representable graphs. 
This paper focuses on $12$-representable graphs. 

Let $[n] = \{1, 2, \ldots, n\}$ for a positive integer $n$. 
A labeled graph $G$ whose labels are drawn from $[n]$ 
is \emph{$12$-representable} if there is a word $w$ over $[n]$ 
such that each letter of $[n]$ appears at least once in $w$ and 
two vertices $i$ and $j$ with $i < j$ are adjacent in $G$ 
if and only if no $i$ occurs before $j$ in $w$. 
In this situation, $w$ is said to \emph{$12$-represent} the graph $G$ 
and $w$ is called a \emph{$12$-representant} of $G$. 
For example, the graph $G_1$ in Figure~\ref{fig:G1}\subref{fig:G1-graph} 
is $12$-representable by a word $w = 8753532847616421$. 
An unlabeled graph $G$ is $12$-representable if 
there is a labeling of $G$ which generates a $12$-representable labeled graph. 

\begin{figure}[ht]
  \begin{minipage}{0.3\textwidth}
    \centering
    \subcaptionbox{\label{fig:G1-graph}}{\begin{tikzpicture}
\def\len{0.8}
\useasboundingbox (-2.3*\len, -1.5*\len) rectangle (2.3*\len, 2*\len);
\tikzstyle{every node}=[draw,circle,fill=white,minimum size=5pt,inner sep=0pt]
\node [label=above:4]        (b) at ($(45:\len*1) + (135:\len*1)$) {};
\node [label=above right:8]  (c) at ($(45:\len*1) + (135:\len*0)$) {};
\node [label=right:6]        (d) at ($(45:\len*1) + (135:\len*-1)$) {};
\node [label=below:7]        (e) at ($(45:\len*0) + (135:\len*-1)$) {};
\node [label=above:1]        (f) at (0, 0) {};
\node [label=above left:5]   (g) at ($(45:\len*0) + (135:\len*1)$) {};
\node [label=left:2]         (h) at ($(45:\len*-1) + (135:\len*1)$) {};
\node [label=below:3]        (i) at ($(45:\len*-1) + (135:\len*0)$) {};
\draw [] 
	(b) -- (c) -- (d)
	(e) -- (f) -- (g)
	(h) -- (i)
	(b) -- (g) -- (h)
	(c) -- (f) -- (i)
	(d) -- (e)
;
\end{tikzpicture}}
  \end{minipage}\hfill
  \begin{minipage}{0.7\textwidth}
    \centering
    \subcaptionbox{\label{fig:G1-model}}{\begin{tikzpicture}
\def\len{0.45}
\useasboundingbox (-1.5*\len, -0.5) rectangle (17.5*\len, 4*\len+0.5);
\tikzstyle{every node}=[minimum size=5pt, inner sep=0pt]
\def\La{4*\len}
\def\Lb{0*\len}
\node [label=left:$L_1$] (L1) at (0, \La) {};
\node [label=left:$L_2$] (L2) at (0, \Lb) {};
\draw [thick] (L1) -- (17*\len, \La);
\draw [thick] (L2) -- (17*\len, \Lb);
\def\pa{(1.5*\len, \La)}
\def\pb{(3.5*\len, \La)}
\def\pc{(5.5*\len, \La)}
\def\pd{(7.5*\len, \La)}
\def\pe{(9.5*\len, \La)}
\def\pf{(11.5*\len, \La)}
\def\pg{(13.5*\len, \La)}
\def\ph{(15.5*\len, \La)}
\node [label=above:1] at \pa {};
\node [label=above:2] at \pb {};
\node [label=above:3] at \pc {};
\node [label=above:4] at \pd {};
\node [label=above:5] at \pe {};
\node [label=above:6] at \pf {};
\node [label=above:7] at \pg {};
\node [label=above:8] at \ph {};
\def\la{(1*\len, \Lb)}
\def\ra{(5*\len, \Lb)}
\def\lb{(2*\len, \Lb)}
\def\rb{(10*\len, \Lb)}
\def\lc{(11*\len, \Lb)}
\def\rc{(13*\len, \Lb)}
\def\ld{(3*\len, \Lb)}
\def\rd{(8*\len, \Lb)}
\def\le{(12*\len, \Lb)}
\def\re{(14*\len, \Lb)}
\def\lf{(4*\len, \Lb)}
\def\rf{(6*\len, \Lb)}
\def\lg{(7*\len, \Lb)}
\def\rg{(15*\len, \Lb)}
\def\lh{(9*\len, \Lb)}
\def\rh{(16*\len, \Lb)}
\node [label=below:1] at \la {};
\node [label=below:1] at \ra {};
\node [label=below:2] at \lb {};
\node [label=below:2] at \rb {};
\node [label=below:3] at \lc {};
\node [label=below:3] at \rc {};
\node [label=below:4] at \ld {};
\node [label=below:4] at \rd {};
\node [label=below:5] at \le {};
\node [label=below:5] at \re {};
\node [label=below:6] at \lf {};
\node [label=below:6] at \rf {};
\node [label=below:7] at \lg {};
\node [label=below:7] at \rg {};
\node [label=below:8] at \lh {};
\node [label=below:8] at \rh {};
\def\op{0.6}
\def\per{50}
\draw [fill=gray!\per, opacity=\op] \pa -- \la -- \ra --  cycle;
\draw [fill=gray!\per, opacity=\op] \pb -- \lb -- \rb --  cycle;
\draw [fill=gray!\per, opacity=\op] \pd -- \ld -- \rd --  cycle;
\draw [fill=gray!\per, opacity=\op] \pf -- \lf -- \rf --  cycle;
\draw [fill=gray!\per, opacity=\op] \ph -- \lh -- \rh --  cycle;
\draw [fill=gray!\per, opacity=\op] \pg -- \lg -- \rg --  cycle;
\draw [fill=gray!\per, opacity=\op] \pe -- \le -- \re --  cycle;
\draw [fill=gray!\per, opacity=\op] \pc -- \lc -- \rc --  cycle;
\end{tikzpicture}}
    \subcaptionbox{\label{fig:G1-model-2}}{\begin{tikzpicture}
\def\len{0.45}
\useasboundingbox (-1.5*\len, -0.5) rectangle (17.5*\len, 4*\len+0.5);
\tikzstyle{every node}=[minimum size=5pt, inner sep=0pt]
\def\La{4*\len}
\def\Lb{0*\len}
\node [label=left:$L_1$] (L1) at (0, \La) {};
\node [label=left:$L_2$] (L2) at (0, \Lb) {};
\draw [thick] (L1) -- (17*\len, \La);
\draw [thick] (L2) -- (17*\len, \Lb);
\def\pa{(1.5*\len, \La)}
\def\pb{(3.5*\len, \La)}
\def\pc{(5.5*\len, \La)}
\def\pd{(7.5*\len, \La)}
\def\pe{(9.5*\len, \La)}
\def\pf{(11.5*\len, \La)}
\def\pg{(13.5*\len, \La)}
\def\ph{(15.5*\len, \La)}
\node [label=above:1] at \pa {};
\node [label=above:2] at \pb {};
\node [label=above:3] at \pc {};
\node [label=above:4] at \pd {};
\node [label=above:5] at \pe {};
\node [label=above:6] at \pf {};
\node [label=above:7] at \pg {};
\node [label=above:8] at \ph {};
\def\la{(6*\len, \Lb)}
\def\ra{(6*\len, \Lb)}
\def\lb{(4*\len, \Lb)}
\def\rb{(12*\len, \Lb)}
\def\lc{(16*\len, \Lb)}
\def\rc{(16*\len, \Lb)}
\def\ld{(3*\len, \Lb)}
\def\rd{(8*\len, \Lb)}
\def\le{(14*\len, \Lb)}
\def\re{(14*\len, \Lb)}
\def\lf{(1*\len, \Lb)}
\def\rf{(1*\len, \Lb)}
\def\lg{(7*\len, \Lb)}
\def\rg{(11*\len, \Lb)}
\def\lh{(9*\len, \Lb)}
\def\rh{(9*\len, \Lb)}
\node [label=below:1] at \la {};
\node [label=below:] at \ra {};
\node [label=below:2] at \lb {};
\node [label=below:2] at \rb {};
\node [label=below:3] at \lc {};
\node [label=below:] at \rc {};
\node [label=below:4] at \ld {};
\node [label=below:4] at \rd {};
\node [label=below:5] at \le {};
\node [label=below:] at \re {};
\node [label=below:6] at \lf {};
\node [label=below:] at \rf {};
\node [label=below:7] at \lg {};
\node [label=below:7] at \rg {};
\node [label=below:8] at \lh {};
\node [label=below:] at \rh {};
\def\op{0.6}
\def\per{50}
\draw [fill=gray!\per, opacity=\op] \pb -- \lb -- \rb --  cycle;
\draw [fill=gray!\per, opacity=\op] \pd -- \ld -- \rd --  cycle;
\draw [fill=gray!\per, opacity=\op] \pg -- \lg -- \rg --  cycle;
\draw [fill=gray!\per, opacity=\op] \pa -- \la -- \ra --  cycle;
\draw [fill=gray!\per, opacity=\op] \pc -- \lc -- \rc --  cycle;
\draw [fill=gray!\per, opacity=\op] \pe -- \le -- \re --  cycle;
\draw [fill=gray!\per, opacity=\op] \pf -- \lf -- \rf --  cycle;
\draw [fill=gray!\per, opacity=\op] \ph -- \lh -- \rh --  cycle;
\end{tikzpicture}}
  \end{minipage}
  \caption{
    \subref{fig:G1-graph} A $12$-representable graph $G_1$. 
    \subref{fig:G1-model} A model of the complement $\overline{G}_1$ of $G_1$. 
    \subref{fig:G1-model-2} Another model of $\overline{G}_1$. 
    The vertices are labeled based on the points on $L_1$. 
    The word $w = 8753532847616421$ obtained from the model in 
    Figure~\ref{fig:G1}\subref{fig:G1-model} is a $12$-representant of $G_1$. 
    We can obtain another $12$-representant $w' = 35278471246$ 
    from the model in Figure~\ref{fig:G1}\subref{fig:G1-model-2}. 
    As mentioned in Example~\ref{ex:G1}, 
    three vertices $2$, $4$, and $7$ are bad in $G_1$; 
    hence, $w'$ is a shortest $12$-representant of $G_1$. 
    }
  \label{fig:G1}
\end{figure}
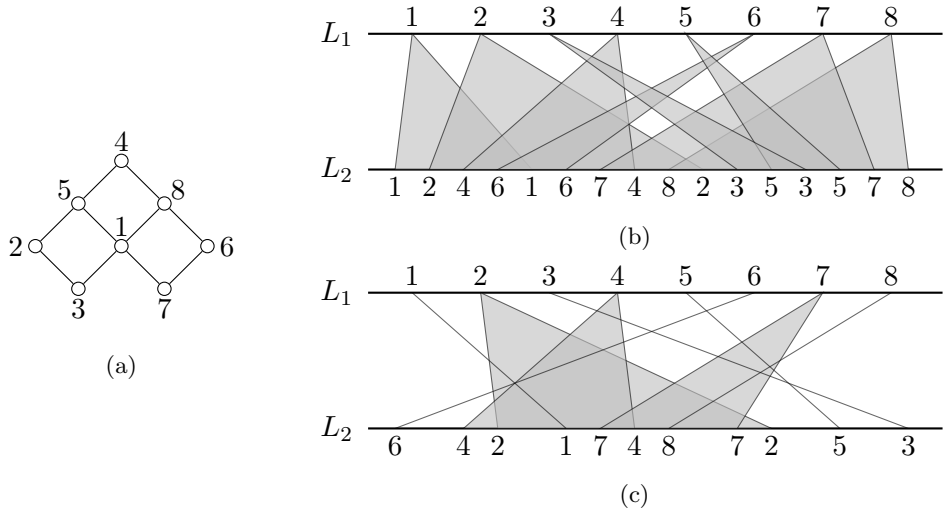

Jones et al.~\cite{JKPR15-EJC} showed that 
the class of $12$-representable graphs is 
a proper subclass of comparability graphs and 
a proper superclass of co-interval graphs and permutation graphs. 
They also provided a characterization of $12$-representable trees 
and a necessary condition for $12$-representability, 
which turned out to be sufficient (Theorem~\ref{thm:good labeling}). 
Chen and Kitaev~\cite{CK22-DMGT} investigated 
the $12$-representability of a subclass of grid graphs 
and presented its characterization. 

The class of $12$-representable graphs is equivalent to 
the complements of simple-triangle graphs~\cite{Takaoka23-DMGT-inpress}. 
This equivalence can be depicted as follows. 
Let $L_1$ and $L_2$ be two horizontal lines in the plane with $L_1$ above $L_2$. 
A point on $L_1$ and an interval on $L_2$ define a triangle between $L_1$ and $L_2$. 
A graph is a \emph{simple-triangle graph}~\cite{CK87-CN} 
if there is a triangle $T_v$ for each vertex $v$ of $G$ 
such that two vertices $u$ and $v$ are adjacent 
if and only if $T_u$ intersects $T_v$. 
The set $\{T_v \colon\ v \in V(G)\}$ of triangles is called 
a \emph{model} or \emph{representation} of $G$; 
we use the term model in this paper to avoid confusion. 
For example, Figure~\ref{fig:G1}\subref{fig:G1-model} is a model of 
the complement $\overline{G}_1$ of the graph $G_1$ in Figure~\ref{fig:G1}\subref{fig:G1-graph}. 
Indeed, two vertices of $G_1$ are adjacent if and only if the corresponding triangles do not intersect. 
Given a model of a simple-triangle graph, we can obtain a $12$-representant of its complement 
by labeling each triangle based on the point on $L_1$ from left to right 
and reading the labels of endpoints on $L_2$ \emph{from right to left}. 
For example, a $12$-representant $w = 8753532847616421$ of $G_1$ can be obtained 
from the model in Figure~\ref{fig:G1}\subref{fig:G1-model}. 
On the other hand, we can construct a model of a simple-triangle graph from a $12$-representant of its complement 
since the complement admits a $12$-representant 
in which each letter appears at most twice (Theorem~\ref{thm:at most twice}). 

It is worth mentioning that a simple-triangle graph admits several models, 
and different models can yield different $12$-representants. 
For instance, Figure~\ref{fig:G1}\subref{fig:G1-model-2} 
illustrates another model of $\overline{G}_1$, 
which provides another $12$-representant $w' = 35278471246$ of $G_1$. 
In addition, triangles can be degenerated to lines 
as in the model of Figure~\ref{fig:G1}\subref{fig:G1-model-2}, 
which correspond to letters appearing twice in the $12$-representant. 

Since simple-triangle graphs can be recognized in $O(nm)$ time~\cite{Takaoka20a-DAM}, 
the equivalence indicates that 
$12$-representable graphs can be recognized in $O(n(\bar{m}+n))$ time~\cite{Takaoka23-DMGT-inpress}, 
where $n$, $m$ and $\bar{m}$ are the number of vertices, edges and non-edges of the given graph, respectively. 
Moreover, a $12$-representant of a graph can be obtained in the same time bound if it exists. 

It should be noted that 
$12$-representable graphs can be recognized in $O(n^2)$ time for labeled graphs, 
i.e., when the labeling is given~\cite{Takaoka23-DMGT-inpress}. 
It is possible that some labeling of a graph admits a $12$-representant 
whereas the other does not (Theorem~\ref{thm:good labeling}). 
Finding valid labeling takes $O(n(\bar{m}+n))$ time, 
but when a valid labeling is given, 
we can obtain the $12$-representant in $O(n^2)$ time (Theorem~\ref{cor:representant}). 

The $12$-representants obtained by the method of~\cite{Takaoka23-DMGT-inpress} 
are of length $2n$, and 
improving the upper bound of the length remains open~\cite{Takaoka23-DMGT-inpress}. 
This is the subject the paper deals with. 
The problem can also be viewed as 
how many triangles in the model could be degenerated to lines. 
The paper proposes an $O(n^2)$-time algorithm 
to compute a shortest $12$-representant of the given \emph{labeled} graph. 
In particular, we show an algorithm to 
transform a $12$-representant $w$ of a labeled graph $G$ 
to a shortest $12$-representant $w'$ of $G$. 
The algorithm is presented in Section~\ref{sec:algo}. 
Section~\ref{sec:preliminaries} introduces some definitions, notations and results used in this paper. 
Section~\ref{sec:conclusion} discusses the unlabeled case and poses an open question. 
Notably, computing a word-representant is NP-hard 
regardless of the labeling since 
the recognition is NP-hard~\cite{Kitaev17-LNCS,KL15-book} and, 
unlike $12$-representable graphs, 
the labeling is not important for word-representable graphs.

\section{Preliminaries}\label{sec:preliminaries}
\paragraph{Graphs}
All graphs in this paper are finite, simple, and undirected. 
We use $uv$ to denote the edge joining two vertices $u$ and $v$. 
For a graph $G$, we use $V(G)$ and $E(G)$ 
to denote the vertex set and the edge set of $G$, respectively. 
We usually denote the number of vertices by $n$. 
The \emph{complement} of a graph $G$ is the graph $\overline{G}$ 
such that $V(\overline{G}) = V(G)$ and 
$uv \in E(\overline{G})$ if and only if $uv \notin E(G)$ 
for any two vertices $u, v$ of $\overline{G}$. 
For a graph $G$, a graph $H$ is an \emph{induced subgraph} of $G$ 
if $V(H) \subseteq V(G)$ and 
$uv \in V(H)$ if and only if $uv \in V(G)$ for all $u, v \in V(H)$. 

A \emph{labeled graph} of a graph $G$ is obtained from $G$ 
by assigning an integer (label) to each vertex. 
A \emph{labeling} of $G$ is an assignment of labels to the vertices of $G$. 
All labels are assumed to be distinct and 
drawn from $[n] = \{1, 2, \ldots, n\}$. 
For a labeled graph, we usually denote its vertices by their labels. 
Unless stated otherwise, graphs are assumed to be unlabeled. 

\paragraph{Trivial upper and lower bounds}
By definition, every $12$-representant contains 
at least one copy of each letter. 
Hence, $n$ is a lower bound of the length of $12$-representants. 
The following theorem yields the upper bound. 
\begin{theorem}[\cite{JKPR15-EJC}]\label{thm:at most twice}
For a $12$-representable graph, 
there is a $12$-representant in which each letter occurs at most twice. 
\end{theorem}

\begin{proposition}
The length of a shortest $12$-representant of a graph is 
at least $n$ and at most $2n$, 
where $n$ is the number of vertices of the graph. 
\end{proposition}

The following theorem can also be used to obtain the lower bound. 
\begin{theorem}[\cite{JKPR15-EJC}]\label{thm:permutation}
A graph is $12$-representable by a permutation if and only if it is a permutation graph. 
\end{theorem}

\begin{corollary}\label{cor:permutation}
If a graph is not a permutation graph, 
then the length of its $12$-representant is at least $n+1$, 
where $n$ is the number of vertices of the graph. 
\end{corollary}

For example, applying Corollary~\ref{cor:permutation} 
to the graph $G_1$ in Figure~\ref{fig:G1}\subref{fig:G1-graph}, 
we obtain the lower bound. 
\begin{example}\label{ex:permutation}
We can see that the graph $G_1$ in Figure~\ref{fig:G1}\subref{fig:G1-graph} 
is not a permutation graph as follows. 
The graph obtained from $G_1$ by removing the vertex $6$ is 
isomorphic to the graph $\overline{\Gamma}_{12}[8]$ in~\cite{Gallai67,MP01-Gallai}. 
Thus, $G_1$ is not the complement of a comparability graph. 
Since any permutation graph is the complement of a comparability graph~\cite{PLE71-CJM}, 
the graph $G_1$ is not a permutation graph. 
Therefore, the length of every $12$-representant of $G_1$ is at least $9$. 
\end{example}

\paragraph{Labeling and recognition}
Labeling is important when dealing with $12$-representable graphs. 
The following theorem indicates that 
not all labeling of a $12$-representable graph 
is $12$-representable. 
\begin{theorem}[\cite{Takaoka23-DMGT-inpress}]\label{thm:good labeling}
A labeled graph $G$ is $12$-representable if and only if 
$G$ contains no induced subgraph $H$ such that 
$\red(H)$ is equal to one of $I_3$, $J_4$, or $Q_4$ in Figure~\ref{fig:I3 J4 Q4}, 
where $\red(H)$ denotes the \emph{reduced form} of $H$, 
i.e., the labeled graph obtained by relabeling $H$ 
so that the $i$-th smallest label is replaced by $i$. 
\end{theorem}

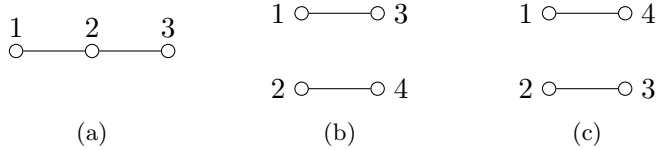
\begin{figure}[ht]
  \centering
  \subcaptionbox{\label{fig:I3}}{\begin{tikzpicture}
\useasboundingbox (-1.5, -.7) rectangle (1.5, 0.7);
\tikzstyle{every node}=[draw,circle,fill=white,minimum size=5pt,inner sep=1pt]
\node [label=above:$1$] (x) at (-1, 0) {};
\node [label=above:$2$] (y) at ( 0, 0) {};
\node [label=above:$3$] (z) at ( 1, 0) {};
\draw [] (x) -- (y) -- (z);
\end{tikzpicture}}
  \subcaptionbox{\label{fig:J4}}{\begin{tikzpicture}
\useasboundingbox (-1.5, -.7) rectangle (1.5, 0.7);
\tikzstyle{every node}=[draw,circle,fill=white,minimum size=5pt,inner sep=1pt]
\node [label= left:$1$] (x) at (-.5,  .5) {};
\node [label=right:$3$] (y) at ( .5,  .5) {};
\node [label= left:$2$] (z) at (-.5, -.5) {};
\node [label=right:$4$] (w) at ( .5, -.5) {};
\draw [] (x) -- (y) (z) -- (w);
\end{tikzpicture}}
  \subcaptionbox{\label{fig:Q4}}{\begin{tikzpicture}
\useasboundingbox (-1.5, -.7) rectangle (1.5, 0.7);
\tikzstyle{every node}=[draw,circle,fill=white,minimum size=5pt,inner sep=1pt]
\node [label= left:$1$] (x) at (-.5,  .5) {};
\node [label=right:$4$] (y) at ( .5,  .5) {};
\node [label= left:$2$] (z) at (-.5, -.5) {};
\node [label=right:$3$] (w) at ( .5, -.5) {};
\draw [] (x) -- (y) (z) -- (w);
\end{tikzpicture}}
  \caption{The labeled graphs $I_3$~\subref{fig:I3}, $J_4$~\subref{fig:J4}, and $Q_4$~\subref{fig:Q4}. }
  \label{fig:I3 J4 Q4}
\end{figure}

It should be noted that the necessity shown in Theorem~\ref{thm:good labeling} 
was first presented by Jones et al.~\cite{JKPR15-EJC}. 
We call a labeling \emph{valid} 
if its resulting graph does not contain 
an induced subgraph isomorphic to $I_3$, $J_4$, or $Q_4$
in the reduced form. 

As mentioned in the introduction, 
$12$-representable graphs are exactly 
the complements of simple-triangle graphs~\cite{Takaoka23-DMGT-inpress}. 
It follows that $12$-representable graphs can be recognized in $O(n(\bar{m}+n))$ time, 
where $\bar{m}$ is the number of non-edges of the given graph. 
However, when a valid labeling is given, 
we can obtain its $12$-representant in $O(n^2)$ time. 
\begin{theorem}[{\cite{Takaoka23-DMGT-inpress}}]\label{cor:representant}
From a valid labeling of a $12$-representable graph $G$, 
a $12$-representant of $G$ can be obtained in $O(n^2)$ time 
without relabeling of $G$. 
\end{theorem}
Therefore, the $12$-representability of \emph{labeled} graphs can be verified in $O(n^2)$ time. 
In more detail, given a valid labeling of a $12$-representable graph $G$, 
we can obtain a model of its complement $\overline{G}$ in $O(n^2)$ time 
by the algorithm in~\cite{Takaoka18-DM}. 
A $12$-representant of $G$ can be obtained from the model 
described in the introduction. 
Hence, the obtained $12$-representant contains 
at most two copies of each letter. 
\begin{theorem}\label{thm:representant}
Given a labeled graph, 
we can test in $O(n^2)$ time whether the graph is $12$-representable. 
If the graph is $12$-representable, then 
its $12$-representant in which each letter occurs at most twice 
can be obtained in $O(n^2)$ time. 
\end{theorem}

\section{Algorithm}\label{sec:algo}
We first improve the lower bound of the length of $12$-representants. 
\begin{definition}
Let $G$ be a labeled graph. 
We refer to a vertex $b$ of $G$ as a \emph{bad vertex} if 
there exist two vertices $a$ and $c$ with $a < b < c$ 
such that $ab, bc \notin E(G)$ and $ac \in E(G)$. 
We call a vertex \emph{good} if it is not a bad vertex.  
\end{definition}

\begin{proposition}\label{prop:bad vertex}
Let $G$ be a $12$-representable labeled graph. 
Each bad vertex must occur twice in every $12$-representant of $G$. 
\end{proposition}
\begin{proof}
Let $w$ be a $12$-representant of $G$, and 
let $a$, $b$, and $c$ be three vertices with $a < b < c$ 
such that $ab, bc \notin E(G)$ and $ac \in E(G)$. 
Since $ac \in E(G)$, every copy of $c$ occurs before the first occurrence of $a$ in $w$. 
Then, $ab, bc \notin E(G)$ implies that 
$b$ occurs after some $a$ and before some $c$ in $w$, 
respectively. 
\end{proof}

Proposition~\ref{prop:bad vertex} leads to the following lower bound. 
\begin{lemma}\label{lemma:lower bound}
Let $G$ be a labeled graph. 
The length of every $12$-representant of $G$ is at least $n + b$, 
where $n$ and $b$ are the number of vertices and bad vertices of $G$, respectively. 
\end{lemma}

In the rest of this section, 
we show that the length of shortest $12$-representants of 
a labeled graph is exactly $n + b$. 
Suppose that there is a $12$-representable graph $G$ and 
its $12$-representant $w$. 
By Theorem~\ref{thm:at most twice}, we can assume 
that each letter occurs at most twice in $w$. 
Proposition~\ref{prop:bad vertex} states that 
all bad vertices occur twice in $w$. 
If the length of $w$ is larger than $n + b$, 
then some good vertices occur twice in $w$. 
Therefore, we propose an algorithm to transform $w$ 
to another $12$-representant of $G$ in which no good vertices occur twice. 

The following is a key observation. 
\begin{proposition}\label{prop:reduce}
Let $G$ be a labeled graph with a $12$-representant of $w$. 
Suppose that a letter $i$ occurs twice in $w$. 
\begin{enumerate}[label=\emph{(\alph*)}]
\item 
Let $j$ be a letter just after the first occurrence of $i$, 
i.e., $w = W_1 i j W_2 i W_3$, 
where $W_1$, $W_2$, and $W_3$ are subwords of $w$. 
If $j < i$ then a word $w' = W_1 j i W_2 i W_3$ is a $12$-representant of $G$. 
\item 
Let $j$ be a letter just before the second occurrence of $i$, 
i.e., $w = W_1 i W_2 j i W_3$, 
where $W_1$, $W_2$, and $W_3$ are subwords of $w$. 
If $j > i$ then a word $w' = W_1 i W_2 i j W_3$ is a $12$-representant of $G$. 
\item 
If two occurrences of $i$ are consecutive in $w$, we can remove one occurrence. 
In other words, if $w = W_1 i i W_2$, 
where $W_1$ and $W_2$ are subwords of $w$, 
then a word $w' = W_1 i W_2$ is a $12$-representant of $G$. 
\end{enumerate}
\end{proposition}
\begin{proof}
(a) Since $j$ occurs before the second occurrence of $i$ in $w$, 
we have $ij \notin E(G)$. 
The letter $j$ still occurs before $i$ in $w'$, 
and hence, $w'$ is a $12$-representant of $G$. 
(b) This can be proved by a similar way. 
(c) Trivial. 
\end{proof}

Proposition~\ref{prop:reduce} leads to Algorithm~\ref{algo}, 
which generates a $12$-representant 
in which no good vertices occur twice. 
The algorithm works under the assumption 
that each letter occurs at most twice in the input. 

\begin{algorithm}[ht]
\caption{Computing a shortest $12$-representant of a labeled graph}\label{algo}
\LinesNumbered
\SetKw{KwDownTo}{downto}
\KwIn{A $12$-representant $w$ of a labeled graph $G$.}
\KwOut{A shortest $12$-representant of $G$.}
\BlankLine
\tcp{We assume that each letter occurs at most twice in $w$.}
\For{$i \gets n$ \KwDownTo $1$}{
  \If{$i$ occurs twice in $w$}{
    Set $p$ to the position of the first occurrence of $i$ in $w$\;
    \lWhile{$w_p > w_{p+1}$}{swap $w_p$ and $w_{p+1}$; $p = p + 1$}
    \lIf{$w_p = w_{p+1}$}{remove $w_p$ from $w$}
  }
}
\For{$j \gets 1$ \KwTo $n$}{
  \If{$j$ occurs twice in $w$}{
    Set $q$ to the position of the second occurrence of $j$ in $w$\;
    \lWhile{$w_q < w_{q-1}$}{swap $w_q$ and $w_{q-1}$; $q = q - 1$}
    \lIf{$w_q = w_{q-1}$}{remove $w_q$ from $w$}
  }
}
\KwRet{$w$}. 
\end{algorithm}

Before proving the correctness of Algorithm~\ref{algo}, 
we see how the algorithm works. 
\begin{example}\label{ex:G1}
Recall that the word $w = 8753532847616421$ is a $12$-representant 
of the graph $G_1$ in Figure~\ref{fig:G1}\subref{fig:G1-graph}. 
Applying Algorithm~\ref{algo} to $w$, 
we obtain the $12$-representant $w' = 35278471246$. 
The operation is illustrated in Figure~\ref{fig:operation}. 
Since three vertices $2$, $4$, and $7$ are bad in $G_1$, 
the word $w'$ is the shortest. 
\end{example}

\begin{figure}[tp]
  \centering\input{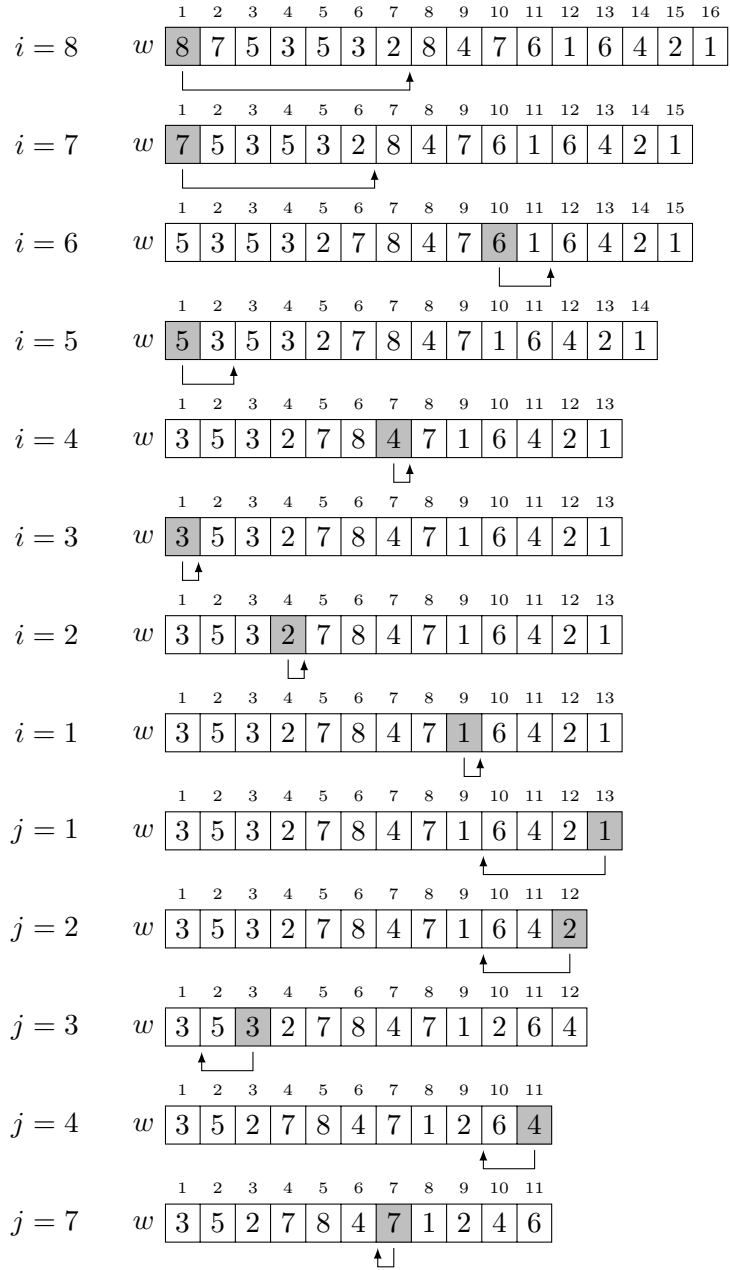}
  \caption{
    The operation of Algorithm~\ref{algo} on the word $w = 8753532847616421$. 
    Each letter is stored in the box whose position appears above. 
    In each iteration, the letter in the shaded box moves, and 
    the arrow denotes the move. 
    The first eight lines illustrate the loops in lines 1--7 
    while the remaining lines illustrate those in lines 8--14. 
    We omitted the cases $j = 5$, $6$, and $8$ 
    because these letters do not appear twice in $w$ at the time. 
  }
  \label{fig:operation}
\end{figure}

Now, we prove the correctness of Algorithm~\ref{algo}. 
\begin{theorem}\label{thm:correctness}
Algorithm~\ref{algo} computes a shortest $12$-representant of the labeled graph 
$12$-represented by the input. 
\end{theorem}
\begin{proof}
Let $G$ and $w$ denote the graph and the input, respectively. 
Proposition~\ref{prop:reduce} ensures that 
the output is still a $12$-representant of $G$. 
By Proposition~\ref{prop:bad vertex}, 
bad vertices occur twice in any $12$-representant. 
Hence, it suffices to prove that no good vertices occur twice in the output. 
We use $w_p$ to denote the letter of $w$ at position $p$, 
i.e., $w = w_1 w_2 \ldots w_{\ell}$. 
\par
Let $j$ be a good vertex of $G$. 
Suppose to the contrary that 
$j$ occurs twice in $w$ at the end of the $j$th loop in lines 8--14. 

\begin{claim}\label{claim:1}
Let $i$ be a letter with $i \geq j$ appearing twice in $w$ 
at the end of the $j$th loop in lines 8--14. 
If $p$ is the first position of $i$, then $w_p < w_{p+1}$. 
\end{claim}
\begin{proof}[Proof of Claim~\ref{claim:1}]
According to Algorithm~\ref{algo}, 
we can see $w_{p'} < w_{p'+1}$ at the end of the $(n-i+1)$th loop in lines 1--7, 
where $p'$ is the first position of $i$ at the time. 
After the $(n-i+1)$th loop, only letters smaller than $i$ move forward. 
Hence, $w_{p''} < w_{p''+1}$ at the beginning of the first loop in lines 8--13, 
where $p''$ is the first position of $i$ at the time. 
Before the end of the $j$th loop in lines 8--14, 
only letters smaller than or equal to $i$ move backward. 
Thus, the claim holds. 
\end{proof}

\begin{claim}\label{claim:2}
Let $k$ be a letter with $k \leq j$ appearing twice in $w$ 
at the end of the $j$th loop in lines 8--14. 
If $q$ is the second position of $k$, then $w_q > w_{q-1}$. 
\end{claim}
\begin{proof}[Proof of Claim~\ref{claim:2}]
According to Algorithm~\ref{algo}, 
we can see $w_{q'} > w_{q'-1}$ at the end of the $k$th loop in lines 8--14, 
where $q'$ is the second position of $k$ at the time. 
After the $k$th loop, only letters larger than $k$ move backward. 
Thus, the claim holds. 
\end{proof}

Let $p$ and $q$ be the first and second positions of $j$, 
respectively, at the end of the $j$th loop in lines 8--14. 
We have $w_{p+1} > w_p$ from Claim~\ref{claim:1} and 
$w_{q-1} < w_q$ from Claim~\ref{claim:2}. 
Let $i_1 = w_{p+1}$ and $k_1 = w_{q-1}$, 
i.e., $w = W_1 j i_1 W_2 k_1 j W_3$, 
where $W_1$, $W_2$, and $W_3$ are subwords of $w$. 
We have $k_1 < j < i_1$ and $k_1j, ji_1 \notin E(G)$. 
Since $j$ is a good vertex, $k_1i_1 \notin E(G)$. 
Hence, some $k_1$ occurs before some $i_1$ in $w$. 
It follows that 
another $i_1$ occurs after $w_{p+1}$ or 
another $k_1$ occurs before $w_{q-1}$. 
\par
Suppose that $i_1$ occurs after $w_{p+1}$. 
Claim~\ref{claim:1} indicates $w_{p+2} > w_{p+1}$. 
Let $i_2 = w_{p+2}$, 
i.e., $w = W_1 j i_1 i_2 W_4 k_1 j W_3$, 
where $W_4$ is a subword of $w$. 
We have $k_1 < j < i_2$ and $k_1j, ji_2 \notin E(G)$. 
Since $j$ is a good vertex, $k_1i_2 \notin E(G)$. 
Hence, some $k_1$ occurs before some $i_2$ in $w$. 
It follows that 
another $i_2$ occurs after $w_{p+2}$ or 
another $k_1$ occurs before $w_{q-1}$. 
\par
On the other hand, suppose that $k_1$ occurs before $w_{q-1}$. 
Claim~\ref{claim:2} indicates $w_{q-2} < w_{q-1}$. 
Let $k_2 = w_{q-2}$, 
i.e., $w = W_1 j i_1 W_5 k_2 k_1 j W_3$, 
where $W_5$ is a subword of $w$. 
We have $k_2 < j < i_1$ and $k_2j, ji_1 \notin E(G)$. 
Since $j$ is a good vertex, $k_2i_1 \notin E(G)$. 
Hence, some $k_2$ occurs before some $i_1$ in $w$. 
It follows that 
another $i_1$ occurs after $w_{p+1}$ or 
another $k_2$ occurs before $w_{q-2}$. 
\par
Continuing in this way, we obtain an infinite sequence 
$w_p < w_{p+1} < w_{p+2} < \cdots$ or 
$w_q > w_{q-1} > w_{q-2} > \cdots$, 
which is a contradiction. 
\end{proof}

From Theorem~\ref{thm:correctness}, we have the main theorem. 
\begin{theorem}
The length of a shortest $12$-representant of a labeled graph is $n + b$, 
where $n$ and $b$ are the number of vertices and bad vertices of the graph, respectively. 
A shortest $12$-representant of a labeled graph can be obtained in $O(n^2)$ time. 
\end{theorem}
\begin{proof}
The proof of Theorem~\ref{thm:correctness} indicates the first statement. 
Theorem~\ref{thm:representant} states that 
a $12$-representant of a labeled graph in which each letter occurs at most twice 
can be obtained in $O(n^2)$ time. 
It is obvious that Algorithm~\ref{algo} takes $O(n^2)$ time. 
Thus, the second statement holds. 
\end{proof}

\section{Concluding remarks}\label{sec:conclusion}
This paper proposes an $O(n^2)$-time algorithm to 
transform a $12$-representant $w$ of a labeled graph $G$ 
to a shortest $12$-representant $w'$ of $G$, 
where $n$ is the number of vertices of $G$. 
This indicates that shortest $12$-representants of labeled graphs can be obtained in $O(n^2)$ time. 
The natural next step is to study the unlabeled case, i.e., 
the problem of finding a shortest $12$-representant of the given unlabeled graph. 

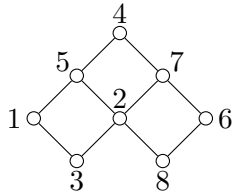
\begin{figure}[ht]
  \centering\begin{tikzpicture}
\def\len{0.8}
\useasboundingbox (-2.5*\len, -1.5*\len) rectangle (2.5*\len, 2*\len);
\tikzstyle{every node}=[draw,circle,fill=white,minimum size=5pt,inner sep=0pt]
\node [label=above:4]        (b) at ($(45:\len*1) + (135:\len*1)$) {};
\node [label=above right:7]  (c) at ($(45:\len*1) + (135:\len*0)$) {};
\node [label=right:6]        (d) at ($(45:\len*1) + (135:\len*-1)$) {};
\node [label=below:8]        (e) at ($(45:\len*0) + (135:\len*-1)$) {};
\node [label=above:2]        (f) at (0, 0) {};
\node [label=above left:5]   (g) at ($(45:\len*0) + (135:\len*1)$) {};
\node [label=left:1]         (h) at ($(45:\len*-1) + (135:\len*1)$) {};
\node [label=below:3]        (i) at ($(45:\len*-1) + (135:\len*0)$) {};
\draw [] 
	(b) -- (c) -- (d)
	(e) -- (f) -- (g)
	(h) -- (i)
	(b) -- (g) -- (h)
	(c) -- (f) -- (i)
	(d) -- (e)
;
\end{tikzpicture}
  \caption{
    Another labeling $G_2$ of the graph $G_1$ in Figure~\ref{fig:G1}\subref{fig:G1-graph}. 
    As shown in~\cite[Theorem 2.18]{CK22-DMGT}, 
    the word $w = 351748246$ is a $12$-representant of $G_2$. 
    }
  \label{fig:G2}
\end{figure}

For $12$-representability, 
labeling matters from the existential point of view (Theorem~\ref{thm:good labeling}). 
The labeling also matters 
to find a shortest $12$-representant of an unlabeled graph. 
In other words, the shortest $12$-representant of some labeling of a graph $G$ 
can be shorter than that of another labeling of $G$. 
For example, as shown in Example~\ref{ex:G1}, 
the shortest $12$-representant of the graph $G_1$ 
in Figure~\ref{fig:G1}\subref{fig:G1-graph} is of length $11 = n + 3$. 
However, if we relabel $G_1$ as $G_2$ in Figure~\ref{fig:G2}, then 
we obtain a $12$-representant $w = 351748246$ of length $9 = n + 1$, 
as shown in~\cite[Theorem 2.18]{CK22-DMGT}; 
only the vertex $4$ is bad in $G_2$. 
We can see from Example~\ref{ex:permutation} that 
$w$ is a shortest $12$-representant of the unlabeled graph. 
Therefore, we conclude this paper by posing the following open question. 
\begin{problem}
Given an unlabeled graph $G$, can we compute 
a valid labeling of $G$ minimizing the number of bad vertices 
in polynomial time? 
\end{problem}

\bibliographystyle{abbrvurl}
\bibliography{ref}
\end{document}